\documentclass[12pt]{article}

\usepackage{fullpage}
\usepackage[T1]{fontenc}
\usepackage{ae,aecompl}
\usepackage[dvips]{graphicx}
\usepackage{amssymb}
\usepackage{amsmath}
\usepackage{subfigure}
\usepackage[round,authoryear]{natbib}
\usepackage{color}
\usepackage{array}
\usepackage{rotating}
\usepackage{colortbl}
\usepackage{tikz}
\usetikzlibrary{patterns}
\newcommand{\Qomit}[1]{}
\newcommand{\Xomit}[1]{}
 \usepackage{dsfont}
\usepackage[normalem]{ulem}
 \usepackage{bbm}
\usepackage{mathabx}
\DeclareMathSymbol{\R}{\mathbin}{AMSb}{"52}

\usepackage{amsthm}

\definecolor{webgreen}{rgb}{0,0.4,0}
\definecolor{webbrown}{rgb}{0.6,0,0}
\definecolor{purple}{rgb}{0.5,0,0.25}
\definecolor{darkblue}{rgb}{0,0,0.7}
\definecolor{darkred}{rgb}{0.7,0,0}
\definecolor{darkgreen}{rgb}{0,0.7,0}
\usepackage[pdfborder=false]{hyperref}
\hypersetup{colorlinks,citecolor=blue,filecolor=black,linkcolor= blue,urlcolor=webgreen,pdfpagemode=none,
pdfstartview=FitH}
\newcommand{\ignore}[1]{}
\newtheorem{lm}{{\sc Lemma}}
\newtheorem{prop}{{\sc Proposition}}
\newtheorem{corollary}{{\sc Corollary}}
\newtheorem{tm}{{\sc Theorem}}
\newtheorem{df}{{\sc Definition}}

\newtheorem{example}{{\sc Example}}

\usepackage{sectsty}
\sectionfont{\centering\normalfont\scshape}
\subsectionfont{\centering\normalfont}

\begin{document} 
 
\begin{titlepage}
\title{ \large{\bf A Belief-Based Characterization of Reduced-Form Auctions}\thanks{The present paper is a substantially revised and extended version of \cite{LA22}. I thank  Eric van Damme,   Debasis Mishra, Rakesh Vohra for their comments
on an earlier draft, and seminar participants for their helpful feedback on several occasions.} }
\author{Xu Lang\thanks{Southwest University of Finance and Economics ({\it langxu@swufe.edu.cn}).} }
 \date{June 27, 2023}
\maketitle

\begin{abstract}
We study games of chance (e.g., pokers, dices, horse races) in the form of agents' first-order posterior beliefs about game outcomes. We ask for any profile of agents' posterior beliefs, is there a game that can generate these beliefs? We  completely characterize all feasible joint posterior beliefs from these games. The characterization enables us to find  a new variant of Border's inequalities  \citep{BO91}, which we call  a belief-based  characterization of Border's inequalities. It also leads to a generalization of Aumann's Agreement Theorem. We show that the characterization results  are powerful in bounding  the correlation of agents' joint posterior beliefs.

\medskip


\end{abstract}
\thispagestyle{empty}
\end{titlepage}

\section{Introduction}
 Consider a card game with three  players,  where each player draws a card from the same deck of cards and hand values are compared to determine the winner.  A natural question is how contradictory can two players 1 and 2' opinions be about a common event $A$ (e.g.,  the event that player 1 wins the game).  The celebrated theorem of Aumann \citep{AU76}  concludes that the players cannot agree to disagree: If the two players' posteriors for  $A$ are common knowledge, then these posteriors must be equal. Suppose now instead of forming opinions  about a common event,  players 1 and 2 are invited to form opinions about two disjoint events $A$ and $B$ (e.g., the events that ones own wins the game).  The question then is whether there is a way to combine their predictions, i.e.,  to come up with a forecast on $A\cup B$. A common intuition suggests that the two players cannot attach to $A$ and $B $ very high probabilities simultaneously.   Then how contradictory can the two players' opinions be? In this paper, we show that a complete solution to this problem can be derived from the well-known Maskin-Riley-Matthews-Border condition (\citealp{MR84}; \citealp{MA84}; \citealp{BO91}) in auction theory.
 
Formally we study a class of games of chance (e.g.,  pokers, dices, horse races) in a mechanism design framework. There is a set of agents. Each agent has a set of possible types (e.g., cards).  The set of all type profiles of agents is associated with a probability space. An outcome function (e.g., the rule of a card game) assigns each {\it true} type profiles of agents a probability of winning for each agent.  Each agent privately observes her type and forms first-order posterior belief about  her {\it payoff-relevant} event, i.e.,  the event that  she wins the game. We call a prior-outcome function pair a game. Given any game, a distribution of joint posterior beliefs of agents would naturally arise from the game.

We completely characterize the set of all feasible joint posterior beliefs that arise from all games.   To obtain a necessary condition for feasibility,  we apply Border's theorem to each game. The key to our analysis is a change of measures, which enables us to find a new variant of the Border inequalities. We call such inequalities belief-based  Border's inequalities (the Border$^*$ inequalities thereafter). To show that the Border$^*$ inequalities is sufficient for feasibility, we use a revelation principle to construct a required underlying game. We next discuss how our main characterization result changes when the prior is fixed. We show that if the prior is independent and atomless, which hold for standard independent private value models, the Border$^*$ inequalities together with posterior independence is characterization of all  feasible joint  posterior beliefs.

To gain some intuition about the Border$^*$ inequalities, we provide three alternative interpretations. The first interpretation is a Border's version of Aumann's Agreement Theorem where agents form opinions about disjoint events in the outcome space.   We show that when  the agents'  joint posteriors  are  common knowledge,  these  posteriors must lie in the probability simplex of the outcome space.  The result reduces to Aumann's  theorem when there are two agents.  The second interpretation is a no-trade condition.   We show that the  Border$^*$ inequalities  correspond  to a no-trade condition for a new class of bets where  agents bet on disjoint events. For a third interpretation, we show that the Border$^*$ inequalities define the core of a coalitional  game.  In this game, agents with various beliefs can form  blocking coalitions and the worth of a coalition is defined by the probability measure of the agents with such beliefs.  
     
 We generalize our main characterization result from games of chance to a general model of information structure design. We assume that there is a state space and different agents can observe signals and form posteriors about possibly different events, i.e., their payoff-relevant events.  We show that if agents' payoff-relevant events are pairwise disjoint, every feasible joint posterior beliefs problem is essentially equivalent to a reduced-form auction problem and the Border$^*$ inequalities give the characterization condition.  Our characterization provides one tractable way to generalize a two-agent characterization of feasible posterior beliefs over a common event in \cite*{ABST21}.\footnote{It is well known that in general the feasibility of joint posterior distributions with many agents is a tough question. } When there are two agents, the agents' posteriors about complementary events reduces to the posteriors about a common event.  

We provide two extensions for our results. First,  we show that our result implies a belief-based characterization of auctions in the form of bidders' first-order posterior beliefs about auction outcomes. For some auctions (e.g.,  knockout auction of a cartel, auctions of arts, all-pay auctions), the auction outcomes may not be publicly observable but  important for the after-auction market interaction. In other markets, insurance companies may sell insurances to bidders for losing the bids. In such cases, bidders' posterior beliefs about auction outcomes will be relevant for the analysis. We show that for independent priors, our main characterization result extends to both Bayesian  and  dominant strategy incentive compatible auctions. In this case, the testing sets in the Border$^*$ inequalities reduce to a smaller class of upper contour sets in the posterior belief space. A similar characterization holds for a fixed prior satisfying independence and absolute continuity.  Our result extends a reduced-form equivalence in \cite{MV10} and \cite*{GGKMS13} from interim allocation outcomes to bidders' posterior beliefs about payoff-relevant events.   

For the general model of designing information structures, we show that Border-like characterization is very powerful in bounding the correlation of feasible joint posterior beliefs.  We find that perfectly positively correlated beliefs are infeasible.  For several classes of classic correlated distributions (i.e., copulas),  positive dependence of beliefs are completely ruled out by the Border$^*$ inequalities. While some existing characterization results  (e.g., \citealp{ABST21}) show that the posteriors for a common event cannot be too negatively correlated,   our result shows that the posteriors for disjoint events cannot be too positively correlated. Hence the two characterizations implied by Aumann's theorem and Border's theorem provide  two ends of the spectrum of beliefs correlation.

\Xomit{
Consider a marriage market where a third party (i.e., an online  platform, family or friends) arranges a marriage for  a woman.\footnote{For example, the family arranged marriage has been the norm in China and many   Asia countries for centuries and is becoming less popular as dating without parental involvement becomes more socially acceptable. } Suppose there are $n$ men available but only one of them can lead to a good match for both sides. Ex ante both the woman and all men as well as the third party do not know the state. Each man can receive some signal from the third party and decide whether to pay a cost (i.e.,  fee for the third party) to approach the woman.     Our results show that in many cases the third party is not able to induce strictly positive correlation  and independent information structure is already the best for the third party. }

In Section 2, we introduce a poker game model and main results. Section 3 discusses the general model.  Section 4 generalizes the main characterization to incorporate incentive compatibility.  Section 5 studies feasible correlation of posterior beliefs.  Section 6 discusses the relations of the paper to the literature. The missing proofs are deferred to Appendix.

   \section{A Poker Game Model}
 Let $N=\{1,\dots,n\}$ with $n\geq 2$ be a set of agents (i.e., gamblers).    Each agent $i$ has a prior probability space  $(T_i, \mathcal{F}_i, \mu_i)$ where $T_i$ is the  type space, with a generic element $t_i\in T_i$,  $\mathcal{F}_i$ is the $\sigma$-algebra and $\mu_i$ is the probability measure.  The space of type profiles is $(T , \mathcal{F} , \mu )$ where $T= T_1\times T_2\times \dots \times T_n$, $\mathcal{F}$ is the
product $\sigma$-algebra, and $\mu$ is a probability measure on the product space  with marginals $\mu_i$. We define $(T_{-i}, \mathcal{F}_{-i}, \mu_{-i} )$ the space of type profiles for agents other than   $i$ analogously.  

An outcome function (i.e., the rule of a poker game) is a measurable function $a: T\to \Delta(N)$ that assigns each type profile of agents a probability of winning for each agent.  Different from a classical mechanism design model, we assume outcomes are determined by {\it true} type profiles. So strategies play no role in a game (i.e., a game of chance).  An interim outcome function  is a measurable function $Q: T\to [0,1]^n$  such that $Q_i$ only depends on $t_i$, i.e., $Q_i(t)=Q_i(t_i) $. That is, an interim outcome function assigns each type  of each agent $i$ an expected probability of winning. We say $Q$ is reduced-form implementable if there exists an outcome function $a$ that implements $Q$, i.e., for all $i\in N$ and $t_i\in T_i$, 
\begin{equation}\label{eq:1}
 Q_i (t_i) =\int_{  T_{-i}} a(i,t_i, t_{-i}) d\mu (t_{-i}|t_i) 
 \end{equation}
 
We call a prior-outcome function pair $((T,\mathcal{F},\mu),a)$  a game. Given a game, we define for each agent $ i$ and type $ t_i $ the  interim winning probability $x_i=Q_i (t_i) $   the (first-order) posterior belief  of agent $i$  about {\it the event  that agent $i $ wins the game}. Then $x_i\in [0,1]$.  For each type profile $t=(t_1,\dots,t_n)$, we say $ (x_1,\dots,x_n)= (Q_1(t_1) ,\dots,Q_n(t_n))$ is a vector of joint posterior beliefs.

We denote by $\nu $ the joint distribution of posterior beliefs   induced by $\mu$ and $a$. That is for every measurable set $  C\subset [0,1]^n$,
      \begin{equation} \label{eq:2}
 \nu (C)=\mu(Q\in C)
    \end{equation}
Then $\nu$ is the pushforward measure of $\mu$ by $Q$ and we write $\nu=Q_{\#} \mu $. Let $\nu_i$ be the marginal probability distribution of  $\nu $. 
    
    \begin{df}  $ \nu  \in \Delta ([0,1]^n)$ is  a feasible joint distribution of posterior beliefs for some game, if there exist a prior probability space $(T,\mathcal{F},\mu)$  and an outcome function   $a$ such that $\nu$ and $ ((T,\mathcal{F},\mu),a)$ satisfy \eqref{eq:1} and \eqref{eq:2}.
         \end{df}
   
    In the following subsections 2.1-2.2 and 2.3, we provide two characterizations of all feasible joint distributions of posterior beliefs from all games: a prior-free characterization and a  fixed-prior characterization.  In both models, we assume the agents have a common prior probability $ \mu$.  While the fixed-prior characterization is    relevant for a game designer whose prior is the common prior, the prior-free characterization is relevant when the agents have a  common prior but the game designer has no such information, in which case a prior-free characterization is very desirable.

\subsection{Main Result}
      
The following Theorem 1 provides a characterization of the set of feasible joint  posterior beliefs from games and is the main result of the paper.
    
  \begin{tm} \label{tm1}  $ \nu  \in \Delta ([0,1]^n)$ is a feasible joint distribution of posterior beliefs for some game  if and only if    
\begin{equation} \label{tm1:1}
  \sum_{i\in N} \int_{  C_ i } x_i  d\nu _i (x_i)\leq 
  \nu \left(\underset{i\in N} {\cup  } (C_i\times [0,1]^{n-1}) \right) 
 \end{equation}
for all  measurable sets  $C_i\subseteq  [0,1] $, $i\in N$ and
\begin{equation} \label{tm1:2}
  \sum_{i\in N} \int_{[0,1] } x_i  d\nu _i (x_i) =1 
   \end{equation}

  \end{tm} 
  
  We call \eqref{tm1:1} and \eqref{tm1:2} the Border$^*$ inequalities. \eqref{tm1:1}  says that for any subsets $C_1,\dots,C_n$ in the agents' posterior belief spaces, the sum of the expectations of beliefs of all agents is no greater than the probability that at least one agent has such a belief.  \eqref{tm1:2} is a new variant of the martingale condition \citep{AM95, KG11}. It says that the sum of the expectations of posterior beliefs of all agents (i.e., the sum of the prior probabilities for different outcomes)  is equal to 1.   
  
  Below we first provide a proof for the Theorem and defer its alternative interpretations to Section 2.2.   The proof of necessity is in essence a reformulation of Border's theorem  with a change of measures.   To show sufficiency, for any posterior beliefs satisfying the characterization inequalities, we explicitly construct a prior probability space and an interim outcome function such that  Border's theorem can be applied. The key to our construction is a revelation argument, i.e., whenever such a feasible solution exists, we can choose the prior probability space equal to the posterior probability space and an interim outcome function equal to the identity map.
       
  \begin{proof} [Proof of Theorem 1]

 Border's theorem  \citep{BO91,BO07} provides a necessary and sufficient condition  for an interim outcome function to be implementable and will be essential for our analysis. We will use a generalization of Border's theorem in  \cite*{CKM13} that allows floor constraint.\\
     
\begin{lm} Let $Q: T\to [0,1]^n $ be measurable and $Q_i(t)=Q_i(t_i) $. $ Q$   is reduced-form implementable if and only if 
 \begin{equation} \label{eq:BO1}
 \sum_{i\in N}  \int_{  E_i} Q_i (t_i) d\mu_i(t_i)   \leq  \mu\left(\underset{i\in N} {\cup  } (E_i\times T_{-i})\right) 
 \end{equation}
   for all measurable sets  $E_i\subseteq T_i$ with $ i\in N$,  and 
  \begin{equation} \label{eq:BO2}
 \sum_{i\in N}  \int_{T_i} Q_i (t_i) d\mu_i(t_i) =1 
 \end{equation}
\end{lm}

{\bf Only If.}   Suppose $  \nu $ is feasible for some game, i.e., there exist  a prior probability space $(T,\mathcal{F}, \mu)$ and an outcome function $a$ such that $\nu$ and $ ((T,\mathcal{F}, \mu),a)$ satisfy \eqref{eq:1} and \eqref{eq:2}.  Then  by definition the interim outcome function $Q$  generated by $ ((T,\mathcal{F}, \mu),a)$  is reduced-form implementable. So by the necessity part of Lemma 1, $Q$ satisfies Border's condition \eqref{eq:BO1} and \eqref{eq:BO2}. 

To obtain a reformulation of Border's theorem, we use a change of measures. By characterization of pushforward measures,   if $Q$ pushes forward  $\mu$ to $\nu$, then for every $\mu$-integrable functions $f: [0,1]^n \to \R$,
 \begin{align}\label{tm:1a}
\int_{[0,1]^n} f(x)d\nu  (x)= \int_{T} f(Q(t))d\mu(t)
 \end{align}
 
 In particular, pick any  measurable set  $C_i\subseteq [0,1] $. Define $E_i=\{t_i\in T_i: Q_i(t_i)\in C_i\}$, then $E_i$  is a measurable set in $T_i$.  Define $f (x)=x_i \cdot\mathbf{1}_{ C_i}(x_i)$ where $\mathbf{1}_{ C_i}$ is the indicator function of $C_i$. Then $f$ depends only on $x_i$.   Substitute $f$ into both sides of \eqref{tm:1a},  we get
    \begin{align}\label{tm:1b}
 \int_{C_i} x_i d\nu_i(t_i)  =  \int_{E_i}  Q_i(t_i)  d\mu_i(t_i) 
 \end{align}

Also note $  \nu $ is feasible implies   $\nu$ and $\mu$ satisfy \eqref{eq:2} for all measurable sets $C\subset [0,1]^n$. In particular, pick $C=\underset{i\in N} {\cup  } (C_i\times [0,1]^{n-1})$, we have 
 \begin{align}\label{tm:1c}
 \nu (C)   =
  \mu\left(Q\in \underset{i\in N} {\cup  } (C_i\times [0,1]^{n-1})\right)  = \mu\left(\underset{i\in N} {\cup  } (E_i\times T_{-i})\right) 
\end{align}
    
Since $Q$ satisfies  Border's condition  for   testing sets $(E_i)_{i \in N}$, substitute \eqref{tm:1b} and \eqref{tm:1c}   into  \eqref{eq:BO1}, we obtain the  inequalities \eqref{tm1:1}.  Finally, for  testing sets $(C_i)_{i \in N}$ with $C_i=[0,1]$, we have $E_i=T_i$. From \eqref{eq:BO2}, we  get the equality \eqref{tm1:2}.

{\bf If.}    Suppose $\nu$ satisfies  \eqref{tm1:1} and   \eqref{tm1:2}. Define a product type space $T=[0,1]^n$ with the Borel $\sigma$-algebra and define a prior probability measure $\mu$ by: for all  measurable $C\subseteq  [0,1]^n$, $\mu(C)=\nu(C) $.   Define $Q:[0,1]^n\to [0,1]^n$ by: for all  $x \in [0,1]^n$,
  \begin{align}
Q_i(x)=x_i, \,\, i\in N 
  \end{align}
  That is, each agent $i$'s  $Q_i$ is equal to her type $x_i$. Since $Q_i$ depends only on $x_i$,  $Q$ is an interim outcome function. Moreover, by construction $\nu$ is the pushforward measure of $\mu$ by $Q$ and \eqref{eq:2} holds.

Now pick any profile of testing sets $(C_i)_{i \in N}$ with $C_i\subseteq [0,1]$ measurable. $\nu$ satisfies  \eqref{tm1:1} for  $(C_i)_{i \in N}$ implies that $Q$ satisfies  \eqref{eq:BO1} for $(E_i)_{i \in N}$ with $E_i=C_i$ for all $i$. From the sufficiency part of Border's theorem, there exists an outcome function $a: T\to \Delta(N)$ such that  $Q$ is the reduced form. Hence we have constructed  a prior  $\mu $   and an outcome function $a$ such that $\nu$,  $\mu $ and $a$ satisfy \eqref{eq:1} and \eqref{eq:2}, i.e.,  $\nu$ is feasible for game $(\mu,a)$.
\end{proof}

  \subsection{Interpretations of Theorem 1}
  
We present three different interpretations of Theorem 1. The first two interpretations by Aumann's agreement theorem and no trade theorem are closely related to the literature, while the third interpretation by blocking is new for our setting. 
     
\noindent{1.  \bf Agree to Disagree.} Below  we  show that Border's inequalities are sufficient to provide a generalized version of Aumann's agreement theorem with $n$ agents. The theorem reduces to the classic Aumann's agreement theorem when $n=2$. In this case, there are two outcomes and the agents' posteriors on complementary events reduces to posteriors on a common event.
      
\begin{tm}  (Aumann's agreement theorem, Border's version)   For each agent $i\in N$, let $E_i$ be the set of types that $ i$ has  posterior $ r_i $ for the event that $i$ wins the game. 

(1) If   $Q_i =r_i$ for all $i\in N$  is common knowledge, then $\sum_{i\in N}r_i=1$.   

(2) For $ |N|=2$,  if $Q_1 =r_1$ and  $Q_2=r_2$  is common knowledge, then $r_1=1-r_2$.

\end{tm}
 Theorem 2 states that if the event $ E_1\times\dots\times E_n$ is  common knowledge, then  the agents' posteriors about their payoff-relevant events must lie in the probability simplex over the outcomes. For intuition, suppose there are three agents and the game designer may wish to persuade each agent that he/she will win for sure.  Can the designer persuade all agents simultaneously? Theorem 2 suggests that  the answer is negative if the agents can communicate, i.e.,  if their posteriors are common knowledge.  
   
\begin{proof}[Proof of Theorem 2] Suppose each   $Q_i $ is a constant $r_i$ on $E_i$. If $E=E_1\times\dots\times E_n$ is common knowledge, then
\begin{equation}\label{eq:3.2}
\mu_i(E_i)=\mu (E)  =\mu\left(\underset{i\in N} {\cup  }(E_i\times T_{-i})\right) \end{equation}
 First apply Border's condition   \eqref{eq:BO1}  to $(E_i)_{i\in N }$,
 \begin{equation}\label{eq:3.3}
 \sum_{i\in N}   \mu_i(E_i) r_i  \leq  \mu\left(\underset{i\in N} {\cup  } (E_i\times T_{-i})\right) 
\end{equation}
Combine \eqref{eq:3.2} and  \eqref{eq:3.3}, we have  
\begin{equation}\sum_{i\in N} r_i\leq 1
\end{equation}
  
Next let $E_i'=T_i\setminus E_i$. Apply Border's inequalities \eqref{eq:BO1}   to  $(E_i')_{i\in N }$, and subtract \eqref{eq:BO2}  from both sides of \eqref{eq:BO1},  we get
\begin{equation} 
 \sum_{i\in N}   \mu_i(E_i) r_i  \geq  \mu (E).
\end{equation}
We get
 \begin{equation}
 \sum_{i\in N} r_i\geq 1
\end{equation}
Hence we conclude  $\sum_{i\in N} r_i=1$.
 \end{proof}

 An immediate implication of Theorem 2 is that  the joint distribution of posteriors with perfectly positive correlation is infeasible. Consider the perfectly positively correlated joint posteriors $(x_1,\dots,x_n)$   uniformly distributed on the main diagonal of $[0,1]^n$, we have the following corollary.

 \begin{corollary}Suppose $\nu\in  \Delta([0,1]^n) $ is uniform on the diagonal $x_1=x_2=\cdots=x_n$, then  $\nu$ is not feasible for any game. \end{corollary}
 
 \begin{proof} The agents' posteriors are common knowledge for all possible realizations $(x_1,\dots,x_n)$, and by Theorem 2 it implies $\sum_i r_i= 1$ for all $r_i\in [0,1]$. But this is impossible for $r_i $ sufficiently large. 
  \end{proof}

\noindent   {  2. \bf A No trade condition. }  Theorem 1 also formalizes a relationship between Border's theorem and No Trade theorem   (\citealp{MS82}, \citealp{MO20}, \citealp{ABST21}).  The Border$^*$ inequalities correspond to a no-trade condition for a new class of bets where agents bet on pairwise disjoint events. Suppose $N=\{1,\dots,n\}$ consists of all possible race outcomes. We define a market for bets where  
each agent  $i\in N$ decides whether to buy one unit of bet $\beta_i$ and each  bet  $\beta_i: N\to \{0,1\}$   is defined by
  \begin{align*}
\beta_i(j)=
\begin{cases}
1 & \textrm{if}~j=i  \\
0 & \textrm{otherwise}
\end{cases}
\end{align*}
 
In other words, each agent  $i $ can bet over whether outcome $i $ occurs and the mediator must clear the market.  If  agent $i $'s posterior for outcome $i$ is $x_i$, then agent $i$ is willing to pay a price of $  x_i$ for the bet $\beta_i$. Suppose further that each agent $i $'s strategy is  to bet if and only if $x_i\in C_i$. For the mediator, she is expected to pay  a price of $1$ if some agent bets and wins.   Condition \eqref{tm1:1} requires that  the mediator's ex ante expected revenue is bound above by the maximum expected payment when at one least one agent bets, for any strategy profile of agents.

\noindent    { 3. \bf A blocking condition.} We next provide an alternative view  on Theorem 1 from  coalitional games where the core of a game is introduced. We first rewrite all  Border$^*$ inequalities in a form with floor constraints. Subtract   \eqref{eq:BO2} from both sides of \eqref{eq:BO1}, we have \eqref{eq:BO1} can be written as
     \begin{equation} \label{tm1:3}
  \sum_{i\in N} \int_{  C_ i } x_i  d\nu _i (x_i)\geq 
  \nu (C_1\times\dots\times C_n)
 \end{equation}
for all  measurable sets  $C_i\subseteq  [0,1] $, $i\in N$.
  
We define a coalitional game problem $(N, N^*, \nu  )$ as follows. Let $N=\{1,\dots,n\}$ be a set of agents. For each agent $i\in N$, let $N_i=[0,1]$ be an   set of subagents of agent $i $.   Let $N^* =\cup_{i\in N }N_i$ be the set of all subagents.  We call $C=\cup_{i\in N }C_i\subseteq N^*$ a coalition.   Let $\nu\in \Delta([0,1]^n) $ be a joint distribution of all profiles of subagents.  The characteristic function $w_{\nu}:  2^{N^*} \to \R$ assigns for each coalition $C \subseteq N^*$ a worth of coalition $w_{\nu}(C)= \nu (C_1\times\dots\times C_n)$. That is  the worth of a coalition  is defined by the probability measure of the coalition. Note that  agents are complement in this game: $C_j=\emptyset$ for some $j $ implies $w_{\nu}(C)=0$.
 
A solution of the problem $(N, N^*, \nu  )$ is  a utility allocation function $\pi: N^*\to [0,1]$ that assigns for each subagent $x_i\in [0,1]$ of each agent $i$ a utility $\pi(x_i)$.  For each coalition $C$, we define the total payoff of coalition $C$ from solution $ \pi$ by
\begin{align}
\pi_{\nu}(C)=  \sum_{i\in N} \int_{  C_ i } \pi(x_i ) d\nu _i (x_i)  
\end{align}
  We say $\pi $ is in the core of the game $(N, N^*, \nu  )$ if for every coalition $C \subseteq N^*$,   $\pi_{\nu}(C)\geq w_{\nu}(C) $ and $\pi_{\nu}(N^*)= w_{\nu}(N^*) $. Then Border$^*$ inequalities \eqref{tm1:2} and \eqref{tm1:3} can be interpreted as that the solution $\pi^*(x_i)=x_i$ is in the core of the game. Alternatively, we can define all core-stable distributions of subagents  by all probability distributions $\nu\in \Delta([0,1]^n)$ such that $\pi^*$ is the core of the game induced by $\nu$.

 \subsection{ Characterization with a given prior}

\Xomit{ We first provide two propositions that give two necessary conditions on feasibility. 

The first proposition states that prior independence implies posterior independence.
    
 \begin{prop}Suppose $ \nu  \in \Delta ([0,1]^n)$ is a feasible joint distribution of posterior beliefs for some game with a given prior $\mu$.   Then $\nu$ is independent if and only if  $\mu$ is independent. 
   \end{prop}
 \begin{proof} Suppose $\mu$ is independent.  From  \eqref{eq:2}, 
 for every measurable sets $  C_i\subseteq [0,1] $, $i\in N$,
      \begin{align*} 
 \nu (C_1\times\dots\times C_n)&=\mu(Q\in C_1\times\dots\times C_n)\\
&= \mu (Q_1\in C_1 ,\dots, Q_n\in C_n )\\
&  =\prod_{i\in N}\mu_i(Q_i\in C_i )\\
&=\prod_{i\in N}\nu_i( C_i )
    \end{align*}
So $\nu$ is independent. Conversely, suppose $\nu$ is independent. A similar argument implies $\mu$ is independent.
           \end{proof}
At first glance Proposition 1 seems counterintuitive as even if the prior is independent, the agents' posteriors about outcomes can be correlated by an auction. For example when all agents directly observe the auction outcome,  their beliefs  about the outcome are perfectly correlated. The reason for independent posteriors in our model is that the interim rules are marginal outcome functions, i.e.,    each agent's interim rule depends only on the agent's type and is separable across different agents. Then the  only way of correlating posteriors is the correlations in the prior.}

 Characterization with a given  prior  is more involved since a feasible joint distribution of posterior beliefs $\nu$ is further constrained by the given prior $\mu$ beyond the reduced form feasibility constraint, i.e.,  there must exist  an interim outcome function $Q$ such that $\nu=Q_{\#}\mu$.  It is immediate that if $\mu$ has an atom, then $\nu$ has an atom.
Furthermore, $\nu$ is not arbitrary. For example, if $\mu$ is a single point mass and $\nu$ consists of two point masses, each
with mass 1/2, then no map can split $\mu$ in half.  Theorem 3 below shows that atoms are almost the only obstruction. Under an atomless assumption about the prior probability spaces we can obtain a characterization similar as Theorem 1.  Let $(T, \mathcal{\mathcal{F}})$ be a measure space such that $\{t\} \in \mathcal{F}$ for every $t\in T$. We say a measure  $\mu$  on $(T, \mathcal{\mathcal{F}})$   is  atomless if for any $A\in \mathcal{F}$ such that $\mu(A)>0$, there exists $B\in \mathcal{F}$ satisfies $0<\mu(B)<\mu(A)$. An example of atomless measures is that $ \mu$ is a continuous probability distribution on $ [0,1]^n$.

  \begin{tm} Let $\mu$ be an independent  and atomless prior.  $ \nu  \in \Delta ([0,1]^n)$ is a feasible joint distribution of posterior beliefs for some game with prior $\mu$   if and only if  $ \nu $   is independent and satisfies the Border$^*$  inequalities \eqref{tm1:1} and \eqref{tm1:2}.
      \end{tm} 
  
 Theorem 3 shows that a fixed prior imposes a strong restriction on the set of feasible posterior beliefs: prior independence implies posterior independence. On the other hand, it shows that a fixed prior imposes  only a weak restriction on the set of feasible posterior beliefs: as long as an atomless prior is considered, all posterior beliefs satisfying the Border$^*$  inequalities are feasible.    
    
  \section{The general model}

 Let $N=\{1,\dots,n\}$, $n\geq 2$ be a finite set of agents and $\Omega=\{ \omega_1,\dots,\omega_m\}$, $m\geq 2$ be a finite state space.  The prior probability of  state $\omega\in \Omega$  is denoted by $p_0^{\omega}$. The agents  have common prior  $p_0$ regarding the states. An information structure $I = ((S_i)_{i\in N}, P)$ consists of  signal spaces   $S_i $ (each equipped with a $\sigma$-algebra)  and a distribution
$P \in \Delta(\Omega \times  S_1\times\dots\times S_n)$, with the marginal of $P$ on  $\Omega$ equal to $p_0$.
 Let $(\omega, s_1,\dots,s_n)$ be a realization.   Then, each agent $i$ observes the signal $s_i\in S_i$.   Let  $q_i(s_i) \in \Delta(\Omega)$ denote  agent $i$'s   posterior belief about $\Omega$   after receiving the signal $s_i$. The posterior belief attributed to $A\subset \Omega$  is given by
 \begin{equation}
q_i^{A}(s_i)= P(A | s_i).
\end{equation}
 Denoted by $P^S$  the marginal of $P$ on the signal space $S$,    and by $p^S_i$ the marginal of $P $ on the signal space $S_i$.

{ \bf Payoff-relevant events.}  Suppose for each agent $i$, there is  a set  $A_i\subseteq \Omega $  of payoff- relevant states for agent $i$. We call $A_i$ the payoff-relevant event for agent $i$.  Below we assume each agent has one payoff-relevant event, i.e., all states in $A_i$ and $\Omega\setminus  A_i$ are payoff equivalent for agent $i$ and agent $i $ is interested in the posterior belief for $A_i$. Our result generalizes to the case where each agent has multiple payoff-relevant events.

Our model cover two important classes of problems: (1) Classical binary states problem with $\Omega=\{\omega_1,\omega_2\} $ and $A_i=\{\omega_1\}$ for all $i\in N$.
 (2) Poker game-like problem where $A_i\cap A_j=\emptyset$ for all $i,j$. That is,  the agents' payoff-relevant events are pairwise disjoint. In particular, we have $\Omega=N$ and  $A_i=\{i\}$ for a poker game model.
 
We denote  $x_i=q_i^{A_i}(s_i)\in [0,1] $ agent $i$'s posterior attributed to $A_i$ and $x=(x_1,\dots, x_n)$.  Denote $[0,1]^n$ the product  space of posterior beliefs for $A_1,\dots,A_n$, and denote by $\psi^I \in  \Delta([0,1]^n)$ the joint distribution of posterior beliefs induced by $I$.   That is, for each measurable set  $B\subset   [0,1]^n$, define
\begin{equation}
\nu^I (B)=P(x\in B).
\end{equation}
   Conversely, we can start with an arbitrary  $ \nu \in    \Delta ([0,1]^n)$  and ask whether it can be generated by some information structure.   
\begin{df}Let $(A_1,\dots,A_n)$ be a profile of payoff-relevant events for agents.
 $ \nu  \in  \Delta (  [0,1]^n)$  is $p_0$-feasible if there exists  some information structure $I$ with prior $p_0$ such that $ \nu =\nu ^I  $.
\end{df}
Below for a probability measure  $\nu  \in \Delta ([0,1]^n)$,  and for agent $i$, we denote by $\nu_i $   the marginal distribution of $\nu$ for agent $i$.

Theorem 4 below characterizes all feasible joint  posterior beliefs  in the general model of designing information structures. The proof of the Theorem also establishes an equivalence between feasibility of reduced-form auctions and feasibility of joint  posterior beliefs. In reduced-form auctions (and poker games), the information structure is defined by the product of outcome function and prior, and  signals are defined by interim outcome functions.  Theorem 4 shows that when the reduced-form auctions with variable priors are considered, they generate the same set of feasible joint posterior beliefs as all information structures in the general model.

  \begin{tm} Suppose $(A_1,\dots,A_n)$ are pairwise disjoint and $\cup_i A_i=\Omega$.
Let  $ \nu  \in \Delta ([0,1]^n)$  is $p_0$-feasible for some $p_0$  if and only if $\nu$ satisfies  \eqref{tm1:1} and \eqref{tm1:2}.
  \end{tm}

\noindent Remark 1. When there are two states and two agents, our Theorem 4 reduces to Theorem 2 of \cite{ABST21}. To see this,  notice that when there are two agents,  the agents' posteriors about complementary events reduce to the posteriors about a common state.

\noindent Remark 2. Theorem 4 can be generalized to a problem where each agent has multiple  payoff-relevant events $A_i^1,\dots,A_i^{l_i}$ and the events of all agents are pairwise disjoint, by a generalized Border theorem for multiple heterogeneous objects (\citealp{ZH21,LY21}), where each agent's reduced-form outcomes are multi-dimensional.

\noindent Remark 3. Theorem 4 restricts attention to a model of joint  posterior beliefs with $A_i\cap A_j=\emptyset$. The same technique can be used to establish  an equivalence between feasibility of reduced-form social choice problems (\citealp{GK11, LM23}) and feasibility of joint  posterior beliefs over a complete set of states, i.e., $A_i^k=\{\omega_k\}$ for all $k=1,\dots,m$ and for all $i$.\footnote{It is worth noting that characterization of reduced-form social choice problems in \cite{GK11} by supporting functions is similar to a  characterization of common prior with many-states in \cite{MO20}. }

\noindent Remark 4.  The general model in Section 3 is not equivalent to the fixed-prior poker game model in Section 2.3. Notice that in the  fixed-prior  poker game   model, the type distributions (i.e., signal distributions) are fixed while only the outcome function (the conditional probabilities of the states given signals) is endogenous.

 \section{Bic-dic equivalence}
 
Our model also provides an approach to auctions in the form of bidders' first-order posterior beliefs about auction outcomes.  In a traditional auction design problem, a bidder can calculate her interim allocation probabilities from different bids in any auction.  Since the interim allocation probabilities influence a bidder's interim utility, it is often understood as a  real variable, i.e., a rule determines real outcomes. The interim allocation probabilities, however, can be also interpreted as a belief variable.  Indeed, a bidder's interim probability of winning  is just her posterior belief about  the event that she wins the object. 
  
 In this section, we study how incentive compatibility further restricts  the set of feasible joint posteriors in auctions. We define an auction model  similar to a poker model in Section 2, except that we let $N=\{1,\dots,n\}$ denote a set of bidders, where bidder 1 is the seller and bidders $2,\dots, n$ are the buyers and  the types profiles in type space $T$ may influence the payoffs of the bidders. Different from a poker game where the outcomes  depend on players' true types, the outcomes in an auction depends on buyers' reported types.    We say an outcome function $a: T\to \Delta(N)$ (i.e., an auction) is (1)  Bayesian 
 incentive compatible (BIC)  if each $Q_i$ is nondecreasing in $t_i$  and (2)
dominant strategy incentive compatible (DIC) if each $a(i)$ is nondecreasing in $t_i$. Our model covers auctions with many buyers and bilateral trade with one seller and one buyer as \cite{MS83}.

We first present an extension  of Theorem 1 for independent priors.   We show that Border$^*$ inequalities, together with monotonicity constraint, characterize all feasible posterior beliefs, for both BIC and DIC auctions. The key to the proof  is using revelation principlem which allows us to construct a BIC interim outcome function. For constructing a DIC auction, we use a reduced form equivalence between BIC and DIC auctions \citep{MV10, GGKMS13}. 
  
 \begin{tm}  Consider all independent priors on $[0,1]^n$.  $ \nu  \in \Delta ([0,1]^n)$ is feasible for some    BIC (or DIC) auction and some prior if and only if $\nu$ is independent and 
   \begin{equation} \label{tm6}
  \sum_{i\in N} \int_{a_i}^1 x_i d \nu _i (x_i)\leq 1-
\prod_{i\in N}  \nu_i(a_i)
 \end{equation}
for all   $a_i\in [0,1]$,  $i\in N$, and \eqref{tm1:2} holds.
  \end{tm}

Theorem 5 states that for independent priors, there is an equivalence between BIC and DIC  in terms of the set of feasible posterior beliefs.  However, this Theorem needs not to imply that  the extreme points of the set of feasible posterior beliefs are generated by the extreme points of reduced-form implementable outcome functions. Below we provide an example to show that, for any two interim outcome functions,  the convex combinations  of their posterior beliefs need not to be the posterior beliefs of these functions' convex combinations.  

\begin{example}
Suppose there are two bidders 1 and 2 with independent and uniform priors on $[0,1]$. Consider auction $a$ that always allocates the object to bidder 1. Then $Q_1^a =1$ and $Q_2^a =0$. The bidders'  posteriors are the Dirac measure $\nu^a =\delta_{(1,0)}$.   Next  consider  auction $b$ that  efficiently allocates the object to the high bidder. Then $Q_i^b(t_i)=t_i$. The bidders' posteriors are $\nu^b(x_1,x_2)=x_1 x_2$.  Finally,   consider auction $c $ as the equal randomization of two auctions $a$ and $b$. Then $Q_1^c(t_1)=\frac{1}{2}+\frac{t_1}{2}$ and $Q_2^c(t_2)=\frac{t_2}{2}$, and  $\nu^c(x_1,x_2)=2(2x_1-1)x_2$ with $x_1\in [ \frac{1}{2},1]$ and $x_2\in [0, \frac{1}{2}]$.  Hence $\nu^c$ is not a convex combination of $\nu^a$ and $\nu^b$.
 \end{example}

We next present an extension of Theorem 3 with  a fixed prior. When both the prior and posteriors are fixed,  the existence of a monotone interim outcome function is nontrivial since the revelation argument no longer works. Below we show that if $\mu$ is an independent and atomless measure on $ [0,1]^n$,  we can construct a monotone interim outcome function that pushes foward $\mu$ to $\nu$.  The proof is based on Brenier's theorem in optimal transport theory (see Lemma 4 in Appendix).
 
 \begin{tm}  Suppose $\mu\in\Delta ([0,1]^n)$ is  independent and absolutely continuous.   $ \nu  \in \Delta ([0,1]^n)$ is feasible for some BIC (or DIC) auction given  $\mu$ if and only if $\nu$ is independent and 
 \eqref{tm1:2} and   $\eqref{tm6}$ hold.
  \end{tm}

\Xomit{
  \subsection{Correlated Priors}

  \begin{tm}  Consider all auction problems with affiliated priors and dominant strategy incentive compatible auctions.  $ \nu  \in \Delta ([0,1]^n)$ is feasible for some problem  if and only if  \eqref{tm1:1} holds
for all $C_i=\{x_i\in [0,1]: x_i\geq a_i\}$ with $a_i\in [0,1]$,  $i\in N$, and \eqref{tm1:2} holds.
  
\end{tm}

 \begin{lm} Let $\mu\in \Delta([0,1]^n) $ be an affiliated  probability distribution.  $ Q$ is reduced-form implementable for a dominant strategy auction  if and only if  \eqref{tm1:1} and  \eqref{tm1:2} hold, and each $Q_i$ is nondecreasing.  
 \end{lm}
 
 \begin{proof} {\bf Only If}  Suppose $\nu$ is affiliated and $a: T\to \Delta(N)$ is a dominant strategy auction. Then $a_i(t_i,t_{-i})$ is nondecreasing in $t_i$.  Define $Q_i(t_i)= \int_{T_{-i}}   a_i(t_i,t_{-i}) d\mu (t_{-i}|t_i)$. By Lemma 3.2.1 in Athey (2001), if $a_i$ is a nondecreasing function in $t_i$ and $\mu$ is affiliated, then $Q_i(t_i)$ is a nondecreasing function in $t_i$. Moreover, from Border's theorem (Lemma 1), we know that $Q$ satisfies \eqref{tm1:1} and  \eqref{tm1:2}.
 
 {\bf If}  Suppose  $Q$ satisfies \eqref{tm1:1} and  \eqref{tm1:2}. $Q$ is reduced-form implementable 
   
 \end{proof}

}
 
\section{Feasible correlation of beliefs }
 
In this section, we further investigate how the Border-like characterization of feasibility in Theorem 4 for the general model imply particular forms of correlation between the agents' posterior beliefs. \cite{ABST21} used their characterization of feasibility to study possible dependence of beliefs. They provided a binary signal example and showed that perfect negative correlation is precluded by their characterization condition. \cite{ZI20} provided a related belief-dependence bound tighter than the usual Fr\'echet bounds. See also \cite{BP19} and
\cite{BP20}  for the references  on coherent opinion of experts, which provided alternative approaches to  tight bounds on the
probability that the pair of posteriors differ by more than a constant delta.  Below we examine a series of examples and show that  the Border-like characterization  impose strong restrictions on the bounds and  possible forms of dependence structures.   Roughly,  the Border$^*$ inequalities require that the posterior  beliefs are not too positively dependent.

\subsection{Independence } For information design problems, there are important concerns that require agents'  posterior beliefs to be (unconditional) independent:  either privacy concerns, or
the avoidance of complicated strategic reasoning on agents’ side, or the endogenous choices by a designer (see e.g. \citet{BBM17}, \citet*{HST21}, \citet{BD21}).   
\cite{ABST21}  Proposition 3 showed that for binary states,  an independent and symmetric distribution is infeasible when the number of agents is sufficiently large. Below we present a similar result for the Border$^*$ inequalities.

 We first show that with independence, a further reduction of  the inequalities in Theorem 4 is possible. Using a similar technique as \cite{CKM13}, we establish the following characterization.

\begin{prop} (1) If $\nu$ is independent, $\nu$  is $p_0$-feasible for some $p_0$  if and only if   \eqref{tm1:2} and \eqref{tm6} hold.
 
(2) If $\nu$ is further symmetric (i.e., $\nu=F^n$), $\nu$  is $p_0$-feasible for some $p_0$  if and only if  
 \begin{equation} \label{co4c}
 \int_a^1 x d F  (x)\leq  \int_a^1 F^{n-1}(x) d F  (x) 
 \end{equation}
for all $a\in [0,1]$ and
  \begin{equation} \label{co4d}
n\int_0^1 x d F  (x)=  1
 \end{equation}
 \end{prop}
 
For independent and symmetric posteriors, Border$^*$ inequalities \eqref{tm1:1} and \eqref{tm1:2}  reduce to one-parameter family of inequalities in Proposition 1. Condition \eqref{co4d} is a symmetric martingale condition. Symmetric posteriors imply that the prior probabilities are symmetric for all states, i.e., $p_0^{\omega_i}=\frac{1}{n}$ for all $i\in N$.   

One might think conditions in Proposition 1 have a structure of majorization inequalities characterized by \cite{HR15} and \cite*{KMS21}, because  Border's theorem falls into this class and \eqref{co4c} follows from Border's theorem. We note that  \eqref{co4c},  however,  is not a majorization inequality in a usual sense. This is because   choice variable $F$ appears on both sides of  \eqref{co4c}. It is still unclear whether the theory of majorization can be applied to this problem.

When all agents have symmetric marginals $F$, from \eqref{co4d} $F$ must vary with $n$. Below we assume agent $1$   is uninformed (e.g., the seller in an auction). We can keep  the marginals of $n-1$ agents fixed and vary the number of agents. 
 
 \begin{prop} Suppose agent $1$ is uninformed, $\nu\in [0,1]^{n-1}$ is independent and has the symmetric marginals $\nu_i=F$  for $i=2,\dots,n$. If the number of agents is sufficiently large ($n\to \infty$),  $\nu$ is not feasible. \end{prop}

To see this result, pick $a_i=0$ for all $ i=2,\dots,n$ and $a_1=1$ in  \eqref{tm6}, we obtain  $(n-1) \mathbb{E}(x)\leq 1$. The inequality must be violated for $n$ sufficiently large. In particular, if $F$ is uniform, it is easy to verify that  condition   \eqref{tm6} holds when $n=3$ and fails for $n\geq 4$, i.e., the independent uniform posteriors are feasible only for three agents.  

\subsection{Correlation }  
 We now discuss several classes  of joint posterior beliefs allowing correlation. Interestingly, Proposition 2 with fixed marginals and variable number of agents generalizes to correlated beliefs.
  
 \begin{prop}Suppose agent $1$ is uninformed and  $\nu\in [0,1]^{n-1}$  has the symmetric marginals $\nu_i=F$  for $i=2,\dots,n$. Then
 
  (1) If the number of agents is sufficiently large ($n\to \infty$),  $\nu$ is not feasible. 
 
(2) In particular if $F$ is uniform, $\nu$ is not  feasible for  $n\geq 4$.
 
 \end{prop}

We next consider the class of three-agent problems where agents 1 and 2 has uniform marginals and agent 3 is uninformed, i.e.,  copulas (see \citealp{NEL07}). We show that the agents' beliefs cannot be too positively dependent. We  introduce the following classic definition of  positive quadratic dependence for bivariate distributions (\citealp{LE66}).
  
 \begin{df}Distribution $\nu\in [0,1]^2$ is positive quadratic dependent if  for all $a_1,a_2\in [0,1]$,
 \begin{equation}
 \nu(X_1\leq a_1,  X_2\leq a_2)\geq   \nu_1(X_1\leq a_1) \nu_2(X_2\leq a_2) 
 \end{equation}
 \end{df}

Consider a class of testing sets which require that for each agent $i $, the posterior belief for $\omega_i$ is greater or equal to $a_i\in [0,1]$:  
  \begin{equation}
C_i =\{ X_i\geq a_i\} 
 \end{equation}
Then the Border$^*$ inequalities in Theorem 4 imply  that for all $a_i\in [0,1] $,   the joint probability for $C_1^c\times  C_2^c$ is bounded above, which gives a bound on the joint distribution of posteriors. 
 
 \begin{prop}Suppose $\nu\in\Delta( [0,1]^2)$   has uniform marginals $\nu_i$ on  $[0,1]$.   If $\nu$ is $ p_0$-feasible for some $p_0$,  then for all $a_1,a_2\in [0,1] $,  
 \begin{equation}\label{eq:p9}
  \nu(X_1\leq a_1, X_2\leq a_2)\leq   \frac{\sum_i a_i^2}{2} 
  \end{equation}
 
 \end{prop}
We can show that for several classes of copulas, the bound  completely rules out positive quadratic dependent beliefs.  For illustration, consider the following 4 classes of copulas:
 
 (1)  $\nu(x_1,x_2)=\min\{x_1, x_2\}$. This copula   corresponds to the Fr\'echet-Hoeffding upper bound,  and in our  case it assumes that agents' beliefs are perfectly positively correlated. Then $\nu$ violates   condition \eqref{eq:p9}. This result verifies that perfect positive correlation is precluded by the no-agreement argument.  

 (2)  $\nu(x_1,x_2)=x_1x_2 +\theta x_1x_2(1-x_1)(1-x_2)$, where $\theta\in [-1,1]$. These correspond to the Farlie-Gumbel-Morgenstern family and is known to be positive quadratic dependent if   $\theta\geq 0$.  Then $\nu$ satisfies condition \eqref{eq:p9} if and only if $\theta\leq 0$, i.e., $\nu$ is independent or negatively dependent.

 (3)  $\nu(x_1,x_2)= (x_1^{-\theta }+x_2^{-\theta }-1)^{-\theta}$, where $\theta\in [-1,+\infty]$. These correspond to the Clayton family and is known to be positive quadratic dependent if   $\theta\geq 0$.  Then $\nu$ satisfies condition \eqref{eq:p9} if and only if $\theta\leq 0$, i.e., $\nu$ is independent or negatively dependent.

 (4)  $\nu(x_1,x_2)=\frac{x_1x_2}{1-\theta(1-x_1)(1-x_2)}$, where $\theta\in [-1,1)$. These correspond to the the Ali–Mikhail–Haq (survival) family and is known to be positive quadratic dependent if   $\theta\geq 0$.  Then $\nu$ satisfies condition \eqref{eq:p9} if and only if $\theta\leq 0$, i.e., $\nu$ is independent or negatively dependent.

For the last example,  we consider that agents 1 and 2 have a joint distribution of posteriors  uniform on the upper triangle
\begin{equation}
\Delta^*=\{(x_1,x_2)\in \R^2_+: x_1+x_2\geq 1,   x_1\leq 1,x_2\leq 1\},
\end{equation}
with a density function
 \begin{center}
 $f(x_1,x_2)=\left\{
\begin{array}{l l}
2 & \text{if}\,\,  (x_1,x_2)\in \Delta^*, \\
0  & \text{otherwise}.\,\,  \\ 
\end{array}%
\right. $
\end{center}
The joint support restriction to the upper triangle captures a notion of positive dependence of beliefs. The next result shows  that this distribution is not feasible.\footnote{Notice that the symmetric marginal density of $\nu $ is given by $f(x)=2x$ for all $x\in[0,1]$. Then the condition in Proposition 4 is given by $2\int_a^1 2x^2 dx\leq  1- (\frac{1}{2}(2a-1)^2) \mathbbm{1}_{a\geq 1/2}$, which is violated for $a\in (0.5, 0.66]$.}

\begin{prop}Suppose $\nu$ is  uniform on $\Delta^*$, then  $\nu$ is not $p_0$-feasible for any $p_0$. \end{prop}

\section{Relation to Literature}
Our paper  contributes to the literature on reduced-form mechanisms. \cite{MR84} and \cite{MA84} first studied reduced-form implementability conditions and \cite{BO91,BO07} proved a conjecture posed by Matthews, which is known as Border's theorem.   \cite{GK11, GK16, GK22}  developed a geometric approach for general social choice problems by support-functions.  \cite*{CKM13} proposed a network flow approach to reduced-form auction with paramodular constraints.  \cite{VO11} characterized reduced forms by polymatroid.  \cite{HR15} and  \cite*{KMS21} characterized symmetric reduced-form auctions by majorization inequalities. \cite{ZH21} and \cite{LY21} generalized reduced-form auctions to multiple goods and constraints.  \cite{LM23}  characterized symmetric reduced form voting rules with two alternatives. \cite*{TVV23} studied a reduced form approach to persuasion where the reduced-form variables are the probabilities with which the receiver takes each of her actions.

Our model provides a new belief-based characterization to reduced-form auctions. We derive a new version of Border's inequalities which characterize all feasible (first-order) posterior beliefs of agents about their payoff-relevant events. Our result formalizes a relationship between Border's theorem and Aumann's agreement theorem. In some sense, Border's theorem is equivalent to Aumann's theorem by appropriately choosing payoff-relevant events.  While our model considers auction-type reduced-form outcomes, our results provide a technique to establish a general equivalence between feasible posterior beliefs problems and reduced-form mechanisms for non-auction problems.

Our model is closely related to the recent literature on feasible  joint posterior beliefs about a common event. When the state space is binary,  \cite{ABST21}  provided a linear characterization of feasible (first-order) posterior beliefs and further proved that these linear inequalities lead  to a quantitative Aumann's agreement theorem. When the state space contains more than two states, \cite{ABST21} and \cite{MO20} provided a no-trade characterization for the problem of feasible joint posterior beliefs.  \cite{LA22} shows that different characterization conditions a la Aumann and Border correspond to different types of market games of bets, which depend on whether the agents making bets on the same events or distinct events and the numbers of buyers and sellers of bets.  While no-trade characterization requires a large number of  inequalities, our Border-like characterization is similar to indicator trading schemes in \cite{ABST21} and provide a tractable way to generalize their indicator trading scheme characterization  from two states to many states. Our paper also contributes to the  literature on common prior and no trade \citep{MS82,  MO94,SM98, FE2000}.   Our result complements the existing result by considering a new class of bets where agents bet on pairwise disjoint events.

The problem of feasible joint posterior beliefs is closely related to the growing literature on Bayesian persuasion started by  \cite{KG11} and  \cite{RS10}.  A special class  of private information is that the information available to each agent reveals nothing about the information available to her opponent. These private signals are literally private and is called  private private information structures   \citep*{HST21}. Private private signals  arise as the worst-case information structure  in some problems of robust auction design (see \cite*{BBM17}; \citealp{BD21}). In auction problems, private private signals are very desirable when the preservation of the buyers' privacy is a major concern.  Our belief-based characterization of Border's inequalities implies that even if there is no robustness concern,  private private signals can arise endogenously in some auction-type information design problems due to feasibility reasons.

\appendix

\section{Appendix}
\label{sec:app1}

    \subsection{Proof of Theorem 3}
     
To prove Theorem 3, we will use the following lemma \citep[Theorem 13.1.1]{DU04}.  Let $(X,\mathcal{F}) $ and $(Y,\mathcal{G})$ be measurable spaces. If $T: X\to Y$ is a measurable map and $\mu$ is a measure on $\mathcal{F}$, then  the pushforward measure     is defined by $T_{\#}\mu(A) =\mu (T^{-1}(A))$ for all $A\in \mathcal{G}$. We have the following lemma.
     
   \begin{lm}
  Let $(X, \mathcal{F}, \mu)$ be an atomless
probability space and let $(Y,\mathcal{G})$ be a Polish space (i.e., separable and  completely metrizable). Then for any probability measure $\nu$ on $(Y,\mathcal{G})$, there exists a measurable function $T: X\to Y$ such that $\nu= T_{\#}\mu$.
  \end{lm}
  \begin{proof}[Proof of Theorem 3]{\bf  Only if.} Suppose  $ \mu$ is an independent and  atomless measure on some type space $(T, \mathcal{\mathcal{F}})$.   The necessity that $\nu$ satisfies the Border$^*$  inequalities   follows from the proofs of Theorem 1.  To show that $\nu$ is independent,   notice that from  \eqref{eq:2} and independence of $\mu$, 
 for every measurable sets $  C_i\subseteq [0,1] $, $i\in N$,
      \begin{align*} 
 \nu (C_1\times\dots\times C_n)&=\mu(Q\in C_1\times\dots\times C_n)\\
&= \mu (Q_1\in C_1 ,\dots, Q_n\in C_n )\\
&  =\prod_{i\in N}\mu_i(Q_i\in C_i )\\
&=\prod_{i\in N}\nu_i( C_i )
    \end{align*}
where the second equality follows from $Q_i(t)=Q_i(t_i)$. Hence $\nu$ is independent.
   
{\bf  If.}  Suppose  $\nu$ is independent and satisfies the Border$^*$  inequalities. Then $\nu=\nu_1\times\dots\times \nu_n$ where each $\nu_i$ is a probability measure on $([0,1] , \mathcal{B})$ where $\mathcal{B}$ is the Borel $\sigma$ -algebra, which is Polish. Since $\mu$ is independent and atomless,  $\mu=\mu_1\times\dots\times \mu_n$ and each $\mu_i$ is atomless. By Lemma 2,  for each agent $i$,  there exists an $\mathcal{F}_i$-measurable function $Q_i: T_i\to [0,1] $ that pushes forward $\mu_i$  to $\nu_i$.  We define an interim  outcome function to be $Q=(Q_1,\dots,Q_n)$. Since $\nu$ satisfies  the Border$^*$  inequalities, by a change of measures we have $Q$ satisfies Border's inequalities. Then Lemma 1 implies that there exists an outcome function $a$ such that $Q$ is the reduced form.  Hence $ \nu$ is feasible for $(\mu,a).$
  \end{proof}

  \subsection{Proof of Theorem 4}

Before presenting the proof, the  following revelation principle  is useful for our analysis. We say an information structure is direct if for each $i$, $S_i=[0,1]  $ and $q_i (s_i)=s_i$ for all $s_i$.

\begin{lm}
 If $  \nu $ is  $p_0$-feasible for some information structure $I_0  $,  then there exists a direct information structure $I$ such that $  \nu $ is  $p_0$-feasible for $I$.
 \end{lm}
 
 \begin{proof}
 If $  \nu $ is  $p_0$-feasible for some $p_0$,   i.e., there exists an  information structure $I_0 =(S,P_0)$  such that $  \nu=  \nu^{I_0}$, then define a direct information structure $I =([0,1]^n, P)$ by $P(\omega, C)=P_0(\omega, q\in C)$ for every measurable $C\subset [0,1]^n$, we have $ \nu=  \nu^{I }$. 
 \end{proof}

 \begin{proof}[Proof of Theorem 4]{\bf  Only if.} Suppose $  \nu $ is feasible. From Lemma 3, there exists a direct information structure  $I =([0,1]^n,P)$ such that $ \nu=\nu^{I }$.  We have
  \begin{align}
   \sum_{i\in N} \int_{C_i}x_i  d\nu_i( x_i)=  &   \sum_{i\in N}  \int_{C_i}x_i d \nu^I_i( x_i) \\
    =&\sum_{i\in N} \int_{C_i  }  q_i^{A_i}(x_i) d \nu^I_i( x_i)\\
  =&\sum_{i\in N} \int_{C_i  }  P(A_i| x_i) d P( x_i) \\
\leq &  \sum_{i\in N}  \int_{\cup_{i\in N} (C_i \times [0,1]^{n-1})} d P(A_i,  x) \\
= & \nu \left(\cup_{i\in N} (C_i \times [0,1]^{n-1}) \right)
\end{align}
where the second equality follows from that $I$ is direct and  the last  equality follows from that  $A_1,\dots,A_n$  is a partition of $\Omega$, and hence  $\sum_{i\in N}P(A_i,  B)= P(B)= \nu(B)$ for all measurable $B\subset [0,1]^n$.
 
{\bf  If.} Suppose $\nu$ satisfies
 \eqref{tm1:1}.   We construct a direct information structure that gives $\nu$ as posteriors. The proof is based on the construction of an  equivalent reduced-form auction problem. 
 Define $D=\{ A_1,\dots, A_n\}$ a set of possible allocations of an object for agents. Define for each agent $i$ type space $ T_i=[0,1] $ with the Borel $\sigma$-algebra. Define $T$ as the product type space and the prior probability measure $\mu=\nu$ on $T$. Define  $Q_i^{A_i}(t)=t_i$.  $\nu$ satisfies \eqref{tm1:1} implies that $Q$ satisfies Border's condition. By Border's theorem, there exists an outcome function $a:T\to \Delta(D)$ such that $Q$ is the reduced form.  
 
 Now define an information structure $(S,P)$ where $ S=T$ and $dP(\omega, t)=\frac{1}{|A_i|}a(A_i,t)d\mu(t)$ for every $\omega\in A_i$ and every $i$.  Then $P\in \Delta(\Omega\times S)$.  By construction
\begin{equation}
 \int_{ T_{-i}}d P(A_i, t_{-i} |t_i)=\int_{  T_{-i}}a(A_i,  t_{-i}, t_i)d\mu_i(  t_{-i}|t_i)= Q_i^{A_i} (t_i)=t_i  
\end{equation}
That is, $Q$ is the posteriors induced by $(S,P)$.  By construction for each $B\subset [0,1]^n$,   \begin{equation}
\nu(B)=  \mu(B)=  P(Q\in B)\end{equation}
Then    $\nu$ is the distribution of posteriors induced by $(S,P)$.    \end{proof}
 
 \subsection{Proof of Theorem 5}

   \begin{proof}[Proof of Theorem 5]{ \bf Only If.} Suppose $ \mu$ is independent, $m$ is a BIC (or DIC) auction and $\nu$ is feasible for the auction problem $(\mu,m)$. A similar analysis as in Theorem 1 implies that $\nu$ satisfies Border$^*$ inequalities \eqref{tm1:1}    and \eqref{tm1:2}.   Also $\mu$  is independent implies that $\nu$ is independent. So  $\nu$ satisfies \eqref{tm6} for all $C_i\subseteq [0,1]$  and hence for all $C_i=\{x_i\in [0,1]: x_i\geq a_i\}$ with $a_i\in [0,1]$.

{ \bf  If.} We  show that when  $\nu$ is independent and satisfies \eqref{tm1:2} and  \eqref{tm6}, then $\nu$ is feasible an for some independent prior and some BIC (or DIC) auction.  The construction is  similar to Theorem 1. Given the type space $T=[0,1]^n$, define a prior probability measure $\mu$ by  $\mu(C)=\nu(C) $ for all  measurable $C\subset  [0,1]^n$.  Define $Q:[0,1]^n\to [0,1]^n$ by  $Q_i(x)=x_i$ for  all  $x \in [0,1]^n$ and $i\in N$. Then $Q$ is an interim outcome function. By construction  $Q$ pushes forward  $\mu$ to $\nu$. Moreover, each $Q_i$ is nondecreasing.

Pick any profile of testing sets $(C_i)_{i \in N}$ with $C_i=\{x_i\in [0,1]: x_i\geq a_i\}$. Since $ Q_i(x_i)=x_i$,    $\nu$ satisfies condition  \eqref{tm6} for  $(C_i)_{i \in N}$ implies that $Q$ satisfies  \eqref{eq:BO1} for $(E_i)_{i \in N}$ with $E_i=\{x_i\in [0,1]: x_i\geq a_i\}$.  From Border's theorem with independent priors and monotonicity constraints  \cite[Corollary 2]{CKM13},  there exists a BIC auction $a: T\to \Delta(N)$ that implements $Q$. From Theorem 2 in \cite{MV10}, we can make $a$ into a DIC auction. Hence we find an independent prior  $\mu $   and a BIC (or DIC) auction $a$ that generate $\nu$, i.e.,   $\nu$ is feasible for a BIC (or DIC) auction problem. \end{proof}

  \subsection{Proof of Theorem 6}

The following theorem in optimal transport establishes the existence of a pushforward (i.e.,a transport plan) with a monotonicity property.

  \begin{lm} Brenier's theorem \citep[Theorem 2.12]{VI03}  If  $\mu$ and  $\nu$ are two probability measures on  $\R^n$  with $ \mu$ absolutely continuous with respect to the Lebesgue measure, then there exists a unique   map $T : [0,1]^n\to [0,1]^n$ that pushes forward $\mu$  to $\nu$ and $T=\nabla \phi$ with $\phi$ convex.
   \end{lm}
   
   \Xomit{
   let  $F$ and $G$ be the distribution functions of $\mu$ and $\nu$, $ G^{-1}\circ F$ is the Brenier map that pushes foward $\mu$ to $ \nu$, where $G^{-1}$
is the generalized inverse of $G$ on $[0, 1]$, i.e., $G^{-1}(a)=  \inf\{x\in \R : G(x)> a\}$. Then $G^{-1}\circ F$ is monotone.}

    \begin{proof}[Proof of Theorem 6]{\bf  Only if.} Suppose  $ \mu $ is a measure on $[0,1]^n$ that is independent and  absolutely continuous.   The necessity that $\nu$ satisfies   \eqref{tm1:2} and  \eqref{tm6} and is independent follows from the proof of Theorem 5.   
   
{\bf  If.}  Suppose  $\nu$ is independent and satisfies  \eqref{tm1:2} and  \eqref{tm6}.   Since $\mu$ is independent and absolutely continuous,   each $\mu_i$ is absolutely continuous. Using Lemma 4 for $n=1$,  for each agent $i$,  there exists a measurable function $Q_i: [0,1]\to [0,1] $ that   pushes forward   $\mu_i$ to $\nu_i$ and $Q_i$ is non-decreasing. We define an interim  outcome function to be $Q=(Q_1,\dots,Q_n)$. Since $\nu$ satisfies \eqref{tm1:2} and  \eqref{tm6}, by a change of variables we have $Q$ satisfies Border's inequalities. Then Lemma 1 implies that there exists an outcome function $a$ such that $Q$ is the reduced form.  Hence $ \nu$ is feasible for $(\mu,a).$
  \end{proof}

\Xomit{

  \section{No-trade interpretation}
  
In this section, we provide an interpretation of Theorem 2  with a  no-trade approach.      \cite{MO20}'s Theorem 3 provides a no-trade characterization (Morris's condition thereafter) that is necessary and sufficient for a joint distribution of the first-order beliefs consistent with a common prior, using the notion of $\Omega$-measurable,  zero-value trades.    Below we show that by restricting zero-value trades to a specific class of bets, Morris's condition further reduces to the Border-like characterization in  Theorem 2.

 We first reformulate Morris's Theorem 3 in   the model with information structures.  Note that Morris considers a coarsened type space where the type space coincides the first-order beliefs and provides a characterization of  a common prior   with the type space  being first-order beliefs.  It  can be seen that the first-order beliefs  consistent  with a common prior correspond to those that can be induced by an information structure.  
 
Following Morris, we consider beliefs with finite supports. Let $F_i\subset \Delta(\Omega)$ be a finite set of posterior beliefs for agent $i$ and denote $F=\times_{i\in N} F_i$. Define $F_i^*$ the projection of $F_i$ onto $\{\Omega_i\}$, i.e.  $F_i^*=\{x_i\in [0,1]: \exists f_i\in F_i\,\,\text{ such that }\,\, x_i=f_i^{\omega_i}\}$. Denote $F^*=\times_{i\in N} F_i^*$.

  \begin{df} {\bf ($\Omega$-measurable zero-value trade)}  (1) A trade is a collection $(x_i)_{i\in N}$ with  each $x_i: F\times \Omega \to \R$.  (2)   A trade $(x_i)_{i\in N}$ is $\Omega$-measurable if $x_i$ is measurable with respect to $(f_i,\omega)$, i.e., $x_i: F_i\times \Omega \to \R$.  (3)   $(x_i)_{i\in N}$ is a zero value trade (for agents 1,\dots,$n$)
if
 \begin{equation}
 \sum_{\omega\in \Omega} f_i^{\omega} x_i(f_i,\omega)=0
 \end{equation}
  for all $f_i\in F_i$ and $i\in N$.
  \end{df}
    
  \begin{lm} (\citealp{MO20}, Theorem 3)   Let $\Omega=\{\omega_0,\omega_1,\dots,\omega_m\} $ and $N=\{1,\dots,n\}$.
  Distribution  $\phi\in\Delta(F)$ is $p_0$-feasible for some $p_0$ if and only if, for every  $\Omega$-measurable, zero-value trade $(x_i)_{i\in N}$,   the following no-trade condition holds
\begin{equation}\label{eq:gen}        
\sum_{f\in F}\phi(f) \max_{\omega \in \Omega} \sum_{i\in N} x_i(f_i, \omega)\geq 0.
\end{equation}
 
  \end{lm}

For binary states,  \cite{MO20} constructs a class of $\Omega$-measurable, zero-value trades   which further reduce condition \eqref{eq:gen} to Theorem 3 in  \cite{ABST21}. For this class of trades,  all of the agents are betting on the same event. We further investigate this approach by considering  agents' bets on the disjoint events to obtain Border's theorem.

Note that  every $\Omega$-measurable, zero-valued  trade  $(x_i)_{i\in N}$  can be written as  $x_i(f_i,\omega)= y_i(f_i,\omega)  a_i(f_i)$  for some $y_i(f_i,\omega) \in \R$ and   $a_i(f_i)\in \R$,  where $y_i$ represents the mediator's bet offering to agent $i$ and $a_i $ denotes  agent $i$'s feasible trading actions given her belief,   including not bet ($a_i=0$), buy  $|a_i|$ units of bet $y_i$ ($a_i>0$) or sell  $|a_i|$ units  of bet $y_i $ ($a_i<0$). Then $y_i$ and $a_i$ can be interpreted as trading strategies of the mediator and the agents. Below we denote  a simple trade by $(x_i)_{i\in N} =((y_i),(a_i))_{i\in N}$. 
  
 We next introduce the notion of simple trades where each of the mediator's bets stipulates only one event, i.e., simple bets.   Simple trades preclude  bets on multiple events.

  \begin{df} {\bf (simple trade)}  An $\Omega$-measurable, zero-value trade $(x_i)_{i\in N} =((y_i),(a_i))_{i\in N}$  is simple, if  $y_i$ satisfies
 $$ y_i (f_i, \omega)= \left\{
\begin{array}{l l}
y_i^1(f_i)   & \text{if}\,\,  \omega\in A_i, \\
y_i^0(f_i)& \text{if}\,\,   \omega\notin A_i,  \\
\end{array}%
\right. $$
 for some $A_i\subset \Omega$ and $y_i^0, y_i^1: F_i\to \R $.
 \end{df}

 Somewhat surprisingly, for the class of simple trades where the agents bet on the distinct states, i.e. $A_i=\{\omega_i\}$,  condition \eqref{eq: gen1} in Lemma 2 then reduces to a characterization similar to the support function description of Border's theorem (see \citealp{GK11}).

\begin{prop} {\bf (Bets on the distinct events)}   Let $\Omega=\{\omega_0,\omega_1,\dots,\omega_n\} $ and $N=\{1,\dots,n\}$.    Distribution  $\phi\in \Delta(F)$ is $p_0$-feasible for some $p_0$  if and only if   for all  $a_i:F_i\to \R$,  $i\in N$,
\begin{equation} \label{eq: gen1}
\sum_{f\in F}\phi(f)(\max\{ 0,   a_1(f_1),\dots,a_n(f_n) \}-\sum_{i\in N}  f_i^{\omega_i}a_i(f_i))\geq 0. 
  \end{equation}

  \end{prop}
  
We note that  Proposition 1 is  different from Theorem 2 as  $\phi$ is defined over some $ F\subset  \Delta(\Omega )^n$ while $\nu$ in Theorem 2 (when supports are finite) is  defined over some $F^* \subset [0,1]^n $.  To obtain this result from Theorem 2, we show that an extension from $F^* $ to $F$ exists.  

\begin{proof} [Proof of Proposition 1](Only If) First suppose that $(x_i)_{i\in N}=((  y_i),(a_i)) _{i\in N}$ with $y_i=( y_i^0, y_i^1)$ is a simple trade.  The zero-valueness and $a_i$ being arbitrary imply that all $( y_i^0, y_i^1)$ in the cone $\{(-\sum_{\omega\in A_i}f_i^{\omega}, 1-\sum_{\omega\in A_i}f_i^{\omega}) \mu$: $\mu\geq 0\}$ are in the same equivalence class of simple trades.  Setting $\mu=1$ is without loss. That is, if  $(x_i)_{i\in N}$ is a simple trade, then  $$ y_i (f_i, \omega)= \left\{
\begin{array}{l l}
1-\sum_{\omega\in A_i}f_i^{\omega}   & \text{if}\,\,  \omega\in A_i, \\
-\sum_{\omega\in A_i}f_i^{\omega}& \text{if}\,\,   \omega\notin A_i, \\
\end{array}%
\right. $$
for some $A_i\subset \Omega$.

Define $N(\omega)=\{i\in N: \omega\in A_i\}$. For each $\omega\in \Omega$ and $f\in F$, we have
\begin{align} 
 \sum_{i\in N}x_i(f_i, \omega)&= \sum_{i\in N(\omega)} (1-\sum_{\omega\in A_i}f_i^{\omega} )a_i(f_i)+   
 \sum_{i\notin N(\omega)}(-\sum_{\omega\in A_i}f_i^{\omega})a_i(f_i)  \notag\\
 &= \sum_{i\in N(\omega)} a_i(f_i)  -\sum_{i\in N} \sum_{\omega\in A_i}f_i^{\omega}a_i(f_i).
  \end{align}
From Lemma 2, substitute the above equality  into condition \eqref{eq:gen} we get the following condition: For all $A_i\subset\Omega$, all  $a_i:F_i\to \R$,  $i\in N$,
\begin{equation} \label{eq: gen2}
\sum_{f\in F}\nu(f)(\max_{\omega\in \Omega} \sum_{i\in N(\omega)} a_i(f_i) -\sum_{i\in N} \sum_{\omega\in A_i}f_i^{\omega}a_i(f_i)) \geq 0. 
  \end{equation}
 
 For every simple trade where all agents bet  on different states, i.e.   $A_i=\{\omega_i\}$, we have
 $$ N(\omega)= \left\{
\begin{array}{l l}
\{i\} & \text{if}\,\,  \omega= \omega_i, \\
 \emptyset  & \text{if}\,\, \omega= \omega_0. \\
\end{array}%
\right. $$
Then from \eqref{eq: gen1} we obtain the condition in the Proposition.

(If) We   show that if $\nu$ satisfies  condition \eqref{eq: gen1}  for all simple trades required in the Proposition, then $\nu$ is $p_0$-feasible for some $p_0$. Define the projection of $\nu$ onto $\Omega^*$ by
 \begin{equation} 
  \nu^{*}(x_1,\dots,x_n )=\nu(f_1^{\omega_1}=x_1,\dots,f_n^{\omega_n}=x_n )
  \end{equation}
 and define $\nu^{*}_i$  the marginal of $\nu^{*}$ for agent $i$. 
Since condition  \eqref{eq: gen1}  holds for  all   $a_i:F_i  \to \R$, $i\in N$,   it implies that for all  $a_i$ measurable with respect to $F_i^*$, in particular for all  $a_i:F_i^*\to \{0,1\}$,  we have
\begin{equation}  \label{C4}
\sum_{i\in N}  x_i a_i(x_i )\nu^*_i (x_i )   \leq \sum_{x\in F^*}\nu^{*} (x)    \max\{ 0,  a_1(x_1),\dots,a_n(x_n) \}.
   \end{equation}
By Theorem 2,  condition \eqref{C4} implies that there  exists a (direct) information structure $P\in \Delta( \Omega\times S)$ with $S =F^*$ such that $P^S= \nu^{*} $. We now extend $P$ to a probability  $\pi\in\Delta(\Omega\times F)$. For any $f^*\in F^*$ and $f\in F$ with $(f_1^{\omega_1}, ..., f_n^{\omega_n})=f^*$,    define $f\setminus  f^*$ the coordinates of $f$ not in $\Omega^*$. Define
 
    \begin{equation}
\mu(f\setminus  f^* )=    \nu(f)/ \nu^{*}(f^*),
    \end{equation}
    and
     \begin{equation}
\pi(\omega, f)=P(\omega,f^*)\mu(f\setminus f^* ).
    \end{equation}
Then we obtain that $\pi$ is the required information structure. This proves that   $\nu $ is $p_0$-feasible for some $p_0$.    
    \end{proof}

We next restrict  simple trades to the market games where agents are   buyers and sellers of bets and the trading unit for each agent is either 0 or 1. The market games are much simpler than simple trades but remain a large class of trades.  Compared to  the simple trades which requires checking an  infinite number of inequalities, the market games require  finitely many inequalities.   
    
   \begin{df}  {\bf (market game)}   A  simple trade $(x_i)_{i\in N}$ with  $(x_i)_{i\in N} =((y_i),(a_i))_{i\in N}$  is a market game, if  (1) for each agent $i \in N$,  the agent is either a buyer ($a_i\geq 0$ for all $f_i\in F_i$) or a seller ($a_i\leq 0$ for all $f_i\in F_i$), and (2) the trade is unitary, i.e.,  $a_i(f_i)\in \{-1,0,+1\}$.  
  \end{df}
  
 The market games generalize  the indicator trading schemes introduced by \cite{ABST21}.
Intuitively, a market game consists a set of pure buyers and sellers of simple bets but not resellers and the trading units of each agent are restricted to be zero or one.  It is immediate that in a market game, if agent $i$ is a buyer, then
 $a_i=\mathbbm{1}_{  C_i}$  for some $C_i \subseteq F_i$, and if  agent $i$ is a seller  then $a_i=-\mathbbm{1}_{  C_i}$  for some $C_i \subseteq F_i$.

The next result shows  that for the class of  market games  with  $n $ buyers and no seller, condition \eqref{eq: gen1} reduces  to the Border-like condition in Theorem 2.

\begin{corollary}  {\bf ($n$-buyers market games)} The condition \eqref{eq:gen} for  market games where each agent $i\in N$ acts as a buyer reduces to the condition in Theorem 2.
 \end{corollary}
 
 For comparison, we also present the result for two-agent binary-state posterior feasibility problem in \cite{ABST21}.
 
\begin{corollary}  {\bf (1 buyer 1 seller market games)} The condition \eqref{eq:gen} for  market games where one of agents  1 and 2 acts as a seller and the other acts as a buyer  reduces to the Border-like condition in \cite{ABST21} Theorem 2.
 \end{corollary}  
}


\begin{thebibliography}{33}
\providecommand{\natexlab}[1]{#1}
\providecommand{\url}[1]{\texttt{#1}}
\expandafter\ifx\csname urlstyle\endcsname\relax
  \providecommand{\doi}[1]{doi: #1}\else
  \providecommand{\doi}{doi: \begingroup \urlstyle{rm}\Url}\fi

\bibitem[Arieli et~al.(2021)Arieli, Babichenko, Sandomirskiy, and
  Tamuz]{ABST21}
I.~Arieli, Y.~Babichenko, F.~Sandomirskiy, and O.~Tamuz.
\newblock Feasible joint posterior beliefs.
\newblock \emph{Journal of Political Economy}, 129\penalty0 (9), 2021.

\bibitem[Aumann(1976)]{AU76}
R.~Aumann.
\newblock Agreeing to disagree.
\newblock \emph{The Annals of Statistics}, pages 1236--1239, 1976.

\bibitem[Aumann and Maschler(1995)]{AM95}
R.~Aumann and M.~Maschler.
\newblock \emph{Repeated games with incomplete information}.
\newblock MIT press, 1995.
 
\bibitem[Bergemann,  Brooks and Morris(2017)Bergemann,  Brooks  and Morris]{BBM17} 
 D.   Bergemann, B. Brooks,  and S. Morris.
 \newblock First-price auctions with general information structures: Implications for bidding and revenue.  
  \newblock \emph{Econonetrica},   85(1):
107-143, 2017.

\bibitem[Border(1991)]{BO91}
K.~Border.
\newblock Implementation of reduced form auctions: A geometric approach.
\newblock \emph{Econometrica}, 59\penalty0 (4):\penalty0 1175--1187, 1991.

\bibitem[Border(2007)]{BO07}
K.~Border.
\newblock Reduced form auctions revisited.
\newblock \emph{Economic Theory}, pages 167--181, 2007.

\bibitem[Brooks and Du(2021)]{BD21}
B.~Brooks and S.~Du.
\newblock Optimal auction design with common values: An informationally-robust
  approach.
\newblock \emph{Econometrica}, 89\penalty0 (2):\penalty0 1313--1360, 2021.
\Xomit{
\bibitem[Brooks et al.(2019) Brooks, Frankel and  Kamenica]{BFK22)}
B. Brooks, A. Frankel, and E. Kamenica
\newblock Information Hierarchies. 
\newblock \emph{Econometrica}, 90\penalty0 (5):\penalty0 2187-2214, 2022.
}
  

\bibitem[Burdzy and Pal(2019)]{BP19}
K.~Burdzy and S.~Pal.
\newblock Contradictory predictions.
\newblock \emph{arXiv preprint arXiv:1912.00126}, 2019.

\bibitem[Burdzy and Pitman(2020)]{BP20}
K.~Burdzy and J.~Pitman.
\newblock Bounds on the probability of radically different opinions.
\newblock \emph{Electronic Communications in Probability}, 25, 2020.

\bibitem[Che et~al.(2013)Che, Kim, and Mierendorff]{CKM13}
Y.~Che, J.~Kim, and K.~Mierendorff.
\newblock Generalized reduced-form auctions: A network-flow approach.
\newblock \emph{Econometrica}, 81\penalty0 (6):\penalty0 2487--2520, 2013.



\bibitem[Dudley(2004)]{DU04}
R. Dudley.
\newblock \emph{Real Analysis and Probability}.
\newblock Cambridge University Press, 2004.

\bibitem[Feinberg(2000)]{FE2000}
Y.~Feinberg.
\newblock Characterizing common priors in the form of posteriors.
\newblock \emph{Journal of Economic Theory}, 91\penalty0 (2):\penalty0
  127--179, 2000.

 
 
 \bibitem[Gershkov et al.(2013)Gershkov, Goeree, Kushnir, Moldovanu, and Shi]{GGKMS13}
A. Gershkov, J. Goeree, A. Kushnir, B. Moldovanu, and X. Shi
\newblock On the equivalence of Bayesian and
  dominant strategy implementation.
\newblock \emph{Econometrica}, 81, 197--220, 2013.

  \bibitem[Goeree and Kushnir(2011)Goeree and Kushnir]{GK11}
J.~Goeree and A.~Kushnir.
 \newblock A Geometric Approach to Mechanism Design. 
  \newblock \emph{Working paper}, December, 2011.
     

\bibitem[Goeree and Kushnir(2016)]{GK16}
J.~Goeree and A.~Kushnir.
\newblock Reduced-form implementation for environments with value
  interdependencies.
\newblock \emph{Games and Economic Behavior}, 99:\penalty0 250--256, 2016.

  \bibitem[Goeree and Kushnir(2022)Goeree and Kushnir]{GK22}
J.~Goeree and A.~Kushnir.
 \newblock A Geometric Approach to Mechanism Design. 
  \newblock \emph{Journal of Political Economy Microeconomics}, 2022.
  
\bibitem[Hart and Reny(2015)]{HR15}
S.~Hart and P.~Reny.
\newblock Implementation of reduced form mechanisms: A simple approach and a
  new characterization.
\newblock \emph{Economic Theory Bulletin}, 3:\penalty0 1--8, 2015.

\bibitem[He et~al.(2021)He, Sandomirskiy, and Tamuz]{HST21}
K.~He, F.~Sandomirskiy, and O.~Tamuz.
\newblock Private private information.
\newblock \emph{Working paper}, December, 2021.
 

\bibitem[Kamenica and Gentzkow(2011)]{KG11}
E.~Kamenica and M.~Gentzkow.
\newblock Bayesian persuasion.
\newblock \emph{American Economic Review}, 101\penalty0 (6):\penalty0 2590- 
  2615, 2011.



\bibitem[Kleiner et~al.(2021)Kleiner, Moldovanu, and Strack]{KMS21}
A.~Kleiner, B.~Moldovanu, and P.~Strack.
\newblock Extreme points and majorization: Economic applications.
\newblock \emph{Econometrica}, 89\penalty0 (4):\penalty0 1557-1593, 2021.

\bibitem[Lang(2022)]{LA22} 
X. Lang.
\newblock
Feasible joint posteriors with many states.
\newblock   https://ssrn.com/abstract=4077632, March, 2022. 

\bibitem[Lang and Yang(2021)]{LY21}
X.~Lang and Z.~Yang.
\newblock Reduced-form allocations for multiple indivisible objects under
  constraints.
\newblock \emph{Working paper}, August, 2021.

\bibitem[Lang and Mishra(2023)]{LM23} 
X. Lang and D. Mishra. 
\newblock
Symmetric reduced form voting, 
\newblock \emph{Theoretical Economics}, forthcoming, April, 2023. 

\bibitem[Lehman(1966)]{LE66}
E.~Lehman.
\newblock Some concepts of dependence.
\newblock \emph{Annals of Mathematical Statistics}, 37:\penalty0 1137-- 1153,
  1966.
  
  \bibitem[Manelli and Vincent(2010)]{MV10}
A. Manelli and D. Vincent.
\newblock Bayesian and dominant-strategy implementation in the independent
  private-values model.
\newblock \emph{Econometrica}, 78, 1905--1938.
  
 


\bibitem[Maskin and Riley(1984)]{MR84}
E.~Maskin and J.~Riley.
\newblock Optimal auctions with risk averse buyers.
\newblock \emph{Econometrica}, 52\penalty0 (6):\penalty0 1473--1518, 1984.

\bibitem[Matthews(1984)]{MA84}
S.~Matthews.
\newblock On the implementability of reduced form auctions.
\newblock \emph{Econometrica}, 52\penalty0 (6):\penalty0 1519--1522, 1984.
 
 \bibitem[Milgrom and Stokey(1982)]{MS82}
P.~Milgrom and N.~Stokey.
\newblock Information, trade and common knowledge.
\newblock \emph{Journal of Economic Theory}, 26\penalty0 (1):\penalty0 17--27,
  1982.
 
\bibitem[Morris(1994)]{MO94}
S.~Morris.
\newblock Trade with heterogeneous prior beliefs and asymmetric information.
\newblock \emph{Econometrica}, 62:\penalty0 1327--1347, 1994.

\bibitem[Morris(2020)]{MO20}
S.~Morris.
\newblock No trade and feasible joint posterior beliefs.
\newblock \emph{Working paper}, July, 2020.


\bibitem[Myerson and Satterthwaite(1983)]{MS83} R. Myerson and
M. Satterthwaite. 
\newblock
Efficient Mechanisms for Bilateral Trading, 
\newblock \emph{Journal of Economic Theory}, 29(2), 265-281, 1983.

\bibitem[Nelson(2007)]{NEL07}
R.~B. Nelson.
\newblock \emph{An Introduction to Copulas}.
\newblock Springer, 2007.

\bibitem[Rayo and Segal(2010)]{RS10}
L. Rayo and I. Segal
\newblock Optimal Information Disclosure.
\newblock \emph{Journal of Political Economy}, 118, 949-987.

 
\bibitem[Samet(1998)]{SM98}
D.~Samet.
\newblock Common priors and separation of convex sets.
\newblock \emph{Games and Economic Behavior}, \penalty0  24(1-2):\penalty0
  172-174, 1998.
  
  
\bibitem[Toikka et al.(2023)Toikka, Vohra, and Vohra]{TVV23}
J. Toikka, A. Vohra, and R. Vohra.
\newblock Bayesian Persuasion: Reduced Form Approach.
\newblock \emph{Journal of Mathematical Economics}, forthcoming, 2023.



\bibitem[Villani(2003)]{VI03}
C. Villani.
\newblock \emph{Topics in optimal transportation}.
\newblock American Mathematical Society,  2003.
 

\bibitem[Vohra(2011)]{VO11}
R.~Vohra.
\newblock \emph{Mechanism Design: A Linear Programming Approach}.
\newblock Cambridge University Press, 2011.

\bibitem[Zheng(2021)]{ZH21}
C.~Zheng.
\newblock Reduced-form auctions of multiple objects.
\newblock \emph{Working paper}, University of Western Ontario, August, 2021.

\bibitem[Ziegler(2020)]{ZI20}
G.~Ziegler.
\newblock Adversarial bilateral information design.
\newblock \emph{Working Paper}, University of Edinburgh, 2020.


\end{thebibliography}
\end{document}